\documentclass[12pt]{article}

\usepackage{theorem}
\usepackage{subfigure}
\usepackage{epsfig}
\usepackage{amssymb}
\usepackage{fullpage}
\usepackage{amsmath}
\usepackage{ifthen}
\usepackage{verbatim}
\usepackage{hyperref}
\usepackage{xspace}
\usepackage{bm}
\usepackage{setspace}
\usepackage{tablefootnote}

\usepackage{hyperref}

\newtheorem{theorem}	 			{Theorem}[section]
\newtheorem{lemma}		[theorem]	{Lemma}

\newtheorem{corollary}		[theorem]	{Corollary}

{\theorembodyfont{\rmfamily} \newtheorem{definition}
[theorem]	{Definition}}
{\theorembodyfont{\rmfamily} }
{\theorembodyfont{\rmfamily} }
{\theorembodyfont{\rmfamily} }
{\theorembodyfont{\rmfamily} }
{\theorembodyfont{\rmfamily} }
{\theorembodyfont{\rmfamily} }
{\theorembodyfont{\rmfamily} }
{\theorembodyfont{\rmfamily} }
{\theorembodyfont{\rmfamily} }
\theoremstyle{break}
{\theorembodyfont{\rmfamily} }

\newenvironment{proof}{\noindent {\em {Proof:}}}{$\blacksquare$\vskip
\belowdisplayskip}

\usepackage[linesnumbered,noend,ruled,noline]{algorithm2e}

\usepackage[title]{appendix}

\usepackage{diagbox}
\usepackage{bm}

\newcommand{\mech}{\mathcal{M}}
\newcommand{\lowermech}{\mathcal{M}_l}
\newcommand{\uppermech}{\mathcal{M}_u}
\newcommand{{\biddingLang}}{\mathcal{B}}

\newcommand{\upperal}{x^u}
\newcommand{\loweral}{x^l}

\newcommand{\OPT}{\texttt{OPT}} 

\usepackage{bm}
\def\vec{\bm}

\def\realspos{\mathbb{R}_{+}}

\DeclareMathOperator*{\argmax}{argmax}

\date{}
\title{When Bidders Are DAOs}
\author{Maryam Bahrani\thanks{a16z crypto. Email:
    \texttt{mbahrani@a16z.com}.} \and Pranav Garimidi\thanks{a16z
    crypto. Email: \texttt{pgarimidi@a16z.com}.} \and Tim
  Roughgarden\thanks{Columbia University \& a16z crypto. 
Author's research at Columbia
  University supported in part by NSF awards CCF-2006737
and CNS-2212745.
Email: \texttt{troughgarden@a16z.com}. 
}}

\begin{document}
\maketitle

\begin{abstract}
    In a typical decentralized autonomous organization (DAO), people organize themselves into a group that is programmatically managed. DAOs can act as bidders in auctions (with ConstitutionDAO being one notable example), with a DAO's bid typically treated by the auctioneer as if it had been submitted by an individual, without regard to any
details of the internal DAO dynamics.

The goal of this paper is to study auctions in which the bidders are DAOs. 
More precisely, we consider the design of two-level auctions in which the "participants" are groups of bidders rather than individuals.  Bidders form DAOs to pool resources, but must then also negotiate the terms by which the DAO's winnings are shared.  We model the outcome of a DAO's negotiations through an aggregation function (which aggregates DAO members' bids into a single group bid) and a budget-balanced cost-sharing mechanism (that determines DAO members' access to the DAO's allocation and distributes the aggregate payment demanded from the DAO to its members).  DAOs' bids are processed by a direct-revelation mechanism that has no knowledge of the DAO structure (and thus treats each DAO as an individual).  Within this framework, we pursue two-level mechanisms that are incentive-compatible (with truthful bidding a dominant strategy for each member of each DAO) and approximately welfare-optimal.

We prove that, even in the case of a single-item auction, the DAO dynamics hidden from the outer mechanism preclude incentive-compatible welfare maximization: No matter what the outer mechanism and the cost-sharing mechanisms used by DAOs, the welfare of the resulting two-level mechanism can be a $\approx \ln n$ factor less than the optimal welfare (in the worst case over DAOs and valuation profiles).  We complement this lower bound with a natural two-level mechanism that achieves a matching approximate welfare guarantee. This upper bound also extends to multi-item auctions in which individuals have additive valuations. Finally, we show that our positive results cannot be extended much further: Even in multi-item settings in which bidders have unit-demand valuations, truthful two-level mechanisms form a highly restricted class and as a consequence cannot guarantee any non-trivial approximation of the maximum social welfare. 
\end{abstract}

\section{Introduction}

In
November~2021, one of the~13 extant original copies of the U.S.\
Constitution was put up for sale by Sotheby's auction house.  Within
days, a DAO (``decentralized autonomous organization'') with over
17000 members formed to crowdsource funds to participate in the
auction.  
This DAO---called ConstitutionDAO, naturally---believed that physical copies of the Constitution should be controlled “by the people,” rather than stuck in private collections. Members of this DAO valued belonging to the collective that wins this auction, and from their perspective the good for sale is therefore non-rivalrous among them
, with a member’s value for winning unharmed by the fact that other members (of the same DAO) win as well.
The members of the DAO were generally anonymous; committed
funds were publicly visible and held in escrow in a smart contract
deployed to the Ethereum blockchain.  All told, the DAO raised roughly
\$47 million leading up to the auction.  Sotheby's sold the copy of
the Constitution using (of course) an ascending auction, with a
designated DAO representative relaying the bids implied by the DAO's
reserves.  Thus, from Sotheby's perspective, ConstitutionDAO was just
like any other bidder, even though in reality it represented the
outcome of coordination of thousands of individuals.
(In the end, ConstitutionDAO lost the auction to Ken
Griffin, CEO of the hedge fund Citadel, and the escrowed funds were
returned to the DAO's participants.)

More generally, the point of a DAO is for people to organize themselves into groups that are programmatically managed. Usually these DAOs are centered around some kind of common cause, for example, forming a social club, collecting art, or funding projects. 
The key innovation that DAOs bring over traditional collectives is that the behavior of the DAO is programmatically enforced via blockchain-secured smart contracts, allowing DAOs to use more complex mechanisms than would traditionally be possible. 
These DAOs may find themselves competing in auctions on behalf of their members. 
In a typical such auction, bids by DAOs might be treated by the seller as
individual bids in (say) a first-price auction, without regard to any
details of the internal group dynamics.

Lest ConstitutionDAO seem like an isolated example, we stress that as blockchains and DAOs become mainstream, this same pattern will likely recur. For example, a DAO concerned with environmental activism could compete in an auction to buy the right to preserve a certain area of land. 
A DAO of musicians could compete for the long-term use of a particular performance venue.
In many of these settings, as long as you are part of the winning DAO, you have a value for the good that is independent of how many other people are part of the DAO. Depending on the application, the good may or may not be excludable (i.e., with the option of excluding select DAO members from access); as our results show, the degree to which the good for sale is excludable 
will have a first-order effect on whether there are mechanisms with good incentive and welfare guarantees.

The goal of this paper is to study auctions in which the bidders are DAOs.
Obvious questions then
include: How should DAO dynamics and internal negotiations be modeled?  Which
key lessons of classical auction theory hold also when bids represent
DAOs, and which ones must be revisited?  Do the mechanism design
problems at the ``upper level'' (the choice of auction) and the ``lower
level'' (the aggregation of preferences of members of a DAO) compose
nicely (\textit{e.g.}, preserving incentive-compatibility), or are there
intricate interactions between them?  Does the aggregation of multiple
individual preferences into a single DAO preference interfere with
natural mechanism design objectives like welfare-maximization, and if
so, by how much?  This paper initiates the study of these questions.

\subsection{Informal Description of Our Model}

Section~\ref{sec:model} details our model; here, we provide an informal
description that is sufficient to understand the overview of results
in the next section.

First, there is an auctioneer that runs an {\em upper-level} mechanism, which
takes as input bids from groups (representing DAOs), and outputs an allocation of items to
groups along with group payments.  This mechanism has no knowledge of
(or is deliberately designed to ignore) the process by which groups'
bids were produced; for all the mechanism knows, each bid was
submitted by an individual bidder. First- and second-price single-item auctions are canonical examples of such mechanisms.

Given the output of the upper-level mechanism, each group must determine
each member's access to then group's winnings, and what to charge each member to cover the overall payment demanded by the
auctioneer.  In the spirit of the Revelation Principle, we model the
result of these determinations as a choice of a (direct-revelation)
budget-balanced cost-sharing mechanism, where the cost to be shared is
the payment demanded by the auctioneer.
In addition to choosing this {\em lower-level} mechanism, a group must decide what to bid (in the
upper-level mechanism), as a function of its members' bids.  We refer to
this mapping (from members' bids to a single group bid) as an {\em
  aggregation function}.

Summarizing, given a choice of upper-level and lower-level mechanisms
(including the aggregation functions), the overall sequence of events
unfolds as follows: (i) each member of each group submits an
individual bid to that group's lower-level mechanism; (ii) each group
maps its members' bids to a group bid via its aggregation function,
which is then submitted to the upper-level mechanism; (iii) the
upper-level mechanism chooses, as a function of the submitted group bids,
an allocation of its items to groups and payments by the groups in
exchange for the allocated items; (iv) each group, having received its
items and a payment request from the upper-level mechanism, uses its
lower-level mechanism to give its members access to the items won and
to share the payment among its members. 

Within this framework, we pursue two goals: (i) dominant-strategy
incentive-compatibility (DSIC); and (ii) (approximate) social welfare
maximization.  By DSIC, we mean a two-level mechanism in which every
member of every group has a dominant strategy, and that strategy is to
bid its true valuation.  For welfare, we consider the valuation
of each group member for the subset of items (its group won
and) to which it has access. The social welfare is the sum of these quantities over all members of all groups.

The core difficulty of this mechanism design problem is aggregating
the preferences of a group and charging payments in an
incentive-compatible way.  For example, if a group always grants all its members
access to all its items and shares the cost evenly,
budget-balanced incentive-compatibility is impossible (due to free
riders underbidding in the hopes that other group members will
shoulder the cost of acquiring a valuable item).  For this reason,
two-level mechanisms with non-trivial guarantees must use lower-level
mechanisms that can exclude group members---presumably the
lower-bidding ones---from accessing some of the items allocated to the
group. For instance, if a DAO of musicians wins access to a performance venue, it may designate a subset of DAO members---intuitively, the members whose bids were actually used to cover the cost of winning the auction---as the only ones eligible to book concerts at the venue.

A second challenge is that distributing payments within a group in
an incentive-compatible way generally precludes the group from
covering payments that match the full welfare of its members.
We prove that this challenge necessarily causes information to be lost
in the aggregation process, which
leads to an unavoidable indistinguishability problem for the upper-level
mechanism and a consequent loss in social welfare.

\subsection{Summary of Results}

We begin with the canonical setting of a single-item auction, and
identify a natural incentive-compatible two-level mechanism that
guarantees an $H_{\ell}$-approximation to the social welfare,
where~$\ell$ denotes the maximum number of bidders in any group and
$H_{\ell} = \sum_{i=1}^{\ell} \tfrac{1}{i} \approx \ln \ell$ the
$\ell$th Harmonic number.  The rough idea is that each group bids the 
maximum amount that can be shared equally among a subset of its members
(subject to individual rationality), with a Vickrey (\textit{i.e.}, second-price) auction serving
as the upper mechanism.  We complement this upper bound with a
matching negative result, assuming only a weak ``equal treatment''
property.\footnote{This property states that if two members of a group
submit identical bids, they should also receive identical allocations
(with either both or neither granted access to the item) and make
identical payments.  This is effectively a symmetry condition on how a
mechanism breaks ties, and it is relevant only for a measure-zero set of
valuation profiles (those in which some valuation is repeated).}
Precisely, every incentive-compatible and individually rational
two-level mechanism that satisfies this property cannot guarantee a
worst-case approximation factor smaller than~$H_{\ell}$.
This lower bound arises from the inability of truthful
mechanisms to elicit payments from groups that match the groups' true
welfare when members of a group have very unequal values.

We then proceed to the multi-item setting. The mechanism of our positive
result for the single-item case extends easily and without degradation to the
setting in which each member of each group has an additive valuation
over items.
However,
the story changes dramatically when we consider the other canonical
``easy case'' for multi-item settings, namely bidders with unit-demand
valuations.  Here, we show that no incentive-compatible and
individually rational two-level mechanism can achieve a
better-than-$n$ approximation of the optimal social welfare (where~$n$ denotes the total number of participants).  The
high-level idea of our proof is to show that incentive-compatibility
is possible in this setting only if there are instances that require
the mechanism to allocate all the items to a single group.  We then
prove that, because the upper mechanism is oblivious to the group
structure (\textit{e.g.}, whether a group represents a single bidder or many),
there will be instances in which such allocations lead to extremely
poor welfare.  This negative result shows that, in particular, the
composition of an incentive-compatible and approximately welfare-maximizing
upper level mechanism (such as the VCG mechanism)
with incentive-compatible and approximately welfare-maximizing lower
mechanisms (such as maximum equal-split cost-sharing) need not lead
to a two-level mechanism with those same properties.
The design of two-level mechanisms with good provable guarantees thus
requires careful coordination between the upper and lower mechanisms,
along with strong restrictions on the set of feasible allocations or the
structure of bidders' preferences.

\subsection{Related Work}
This paper follows in the tradition of a long line of works, beginning
with~\cite{nisan_ronen} and~\cite{lehmann_ocallaghan_shomam}, that
study the problem of incentive-compatible approximate
welfare-maximization under side constraints.  (Without side
constraints, exact incentive-compatible welfare-maximization can be
achieved using the VCG mechanism.)  Like most of these works, we focus
on a prior-free setting, worst-case (over valuation profiles) relative
approximation of the optimal welfare, and dominant-strategy
incentive-compatibility.

Many of the papers in this line belong to the field of algorithmic
mechanism design, in which the side constraints impose bounds on the
amount of computation or communication used by a mechanism (see
e.g.~\cite{roughgarden_talgamcohen}).  For example, in multi-item
(combinatorial) auctions, the size of a bidder's valuation (and hence
the communication used by a direct-revelation mechanism) is generally
exponential in the number of items~$m$.  If a mechanism uses an amount
of communication that is bounded by a polynomial function of~$m$,
bidders will be unable to report fully their valuations and the
mechanism must ultimately make its allocation (and payment) decisions
with incomplete information.  Unsurprisingly, full
welfare-maximization is generally impossible in such settings, even
after setting aside any incentive-compatibility constraints.

Somewhat similarly, in the two-level mechanism framework studied in
this paper, one of the side constraints requires the upper mechanism
to base its allocation (and payment) decisions on incomplete
information (group bids, rather than the individual member bids that
led to those group bids), again precluding any direct-revelation
solution.  Here, however, it is the combination of limited information
{\em and} the incentive-compatibility constraint that rules out exact
welfare-maximization.\footnote{\textit{E.g.}, in a single-item setting, exact
  welfare maximization (without incentive-compatibility) is easy to
  achieve: bidders bid truthfully, each group reports the sum of its
  members' bids, the upper mechanism chooses the highest group bid and
  charges that group its bid, and the winning group then charges all
  its members their bids.} Accordingly, the crux of our lower bound
proofs is to delineate the limitations of incentive-compatible
(two-level) mechanisms, not to identify any intrinsic difficulty of
the underlying optimization problem.  Several papers in algorithmic
mechanism design, such as \cite{papadimitriou_schapira_singer} and
\cite{dobzinski_vondrak_ec12}, face similar challenges when proving
that there are welfare-maximization problems for which
incentive-compatibility constraints degrade the best-possible
approximation factor achievable by a polynomial-time algorithm.
Results in the spirit of Roberts's Theorem~\cite{roberts}, which
characterize the set of incentive-compatible mechanisms for a given
setting, can also have a similar flavor.

Our model is similar to the one in \cite{rachmilevitchauctions}, although that paper has very different goals than the present work. Rather than considering the space of dominant-strategy incentive-compatible mechanisms, as we do here, the paper \cite{rachmilevitchauctions} fixes specific (non-incentive-compatible) auctions for the upper mechanism and aggregation rules and cost-sharing rules in the lower mechanism before characterizing (prior-dependent) equilibrium strategies for the participants. The analysis in \cite{rachmilevitchauctions} is also restricted to the specific setting in which a single group of bidders competes with a single individual bidder. 

The lower mechanisms in our two-level framework are required to be
budget-balanced cost-sharing mechanisms, and there is a large literature on such mechanisms. Naturally, some ideas in our proofs also have precursors in that literature; a few other papers that bear resemblance to the present work are \cite{feigenbaum2000sharing}, \cite{moulin2001strategyproof},
\cite{quantifying_inefficiency}, and \cite{shapley}.  There are two
major differences, however, between the use of cost-sharing mechanisms
in our framework and the settings in which they are traditionally
studied.  
First, in the standard setup, a cost-sharing mechanism chooses an 
outcome that incurs a cost (\textit{e.g.}, the cost of building a bridge) and
the goal is to maximize the social welfare (the total value of the
winning participants for the chosen outcome, minus the cost of that
outcome) or minimize the social cost (the cost of the chosen outcome,
plus the total value of the losing participants).  The cost of an
outcome in this setup is exogenously specified, independent of
participants' bids.  In our two-level framework, the ``cost'' to be
shared is an endogenously specified transfer (from a group to the
auctioneer, as a function of other groups' bids) that does not detract
from the social welfare.  In the traditional setup,
incentive-compatible budget-balanced cost-sharing mechanisms cannot
guarantee any approximation of the optimal social welfare for even the
simplest of problems, and for this reason the social cost objective is
usually considered instead~\cite{quantifying_inefficiency}.  Here,
with costs internal rather than external to the system, a non-trivial
approximation to the social welfare objective is possible (\textit{e.g.}, for
single-item settings).
Second, cost-sharing mechanisms are traditionally studied as
stand-alone direct-revelation mechanisms, whereas here they constitute
one component of a more complex mechanism.  Our results show that
plugging incentive-compatible cost-sharing mechanisms into a two-level
mechanism with an incentive-compatible upper mechanism does not
generally preserve incentive-compatibility (see
Appendix~\ref{app:nocompose}).  This lack of modularity
between the upper and lower mechanisms suggests that the power and
limitations of two-level mechanisms must be studied from first
principles.

Resembling our two-level framework is a sequence of papers on bidding rings and collusion in auctions, namely \cite{graham1987collusive}, \cite{marshall2007bidder}, \cite{leyton2000bidding}, \cite{leyton2002bidding}, and \cite{mcafee1992bidding}. In this line of work, a group of bidders participates in a first- or second-price single-item auction by coordinating among themselves in a \textit{bidding ring} to increase their expected utility. 
Unlike in our model, in which individuals value belonging to the winning group (and hence many can ``win''), in these papers there is only one winning individual.
As a result, these works require transfers within the bidding ring to incentivize agents to join. They also require individuals to have a common prior and compute (non-dominant) equilibrium strategies.
Finally, these works do not consider the welfare loss due to collusion, as we do here.

Also reminiscent of our two-level framework but more distantly related are
works that consider various mechanism design setups with
intermediaries.  For example, \cite{babaioff2016mechanism} consider
facility location problems on trees and assume that strategic agents
report to mediators that then act on their behalf.  In addition to
studying a very different underlying optimization problem, a key
aspect of this paper is that mediators are assumed to be strategic
(whereas the analog in our framework, the lower mechanisms, have no
agency).  A related line of research, motivated by online
advertisement exchange systems, considers auctions in which bidders
report bids to intermediaries who in turn submit bids to a seller
(\textit{e.g.}, \cite{intermediaries} and
\cite{AggarwalBGP22}).  In addition to focusing on strategic
intermediaries, these works are primarily concerned with approximate
revenue-maximization (as opposed to approximate welfare-maximization).

\section{Preliminaries}\label{sec:model}

We consider a setting where an auctioneer is selling a set of items,
$[m]=\{1,...,m\}$ to $k$ distinct groups. We denote group $j$ by $G^j$
and let $G^j$ have $n_j$ bidders, with a total of $n$ bidders across all
of the groups. We only allow bidders to be part of a single group and
assume that the auctioneer only interacts with a group as a whole, with
no insight into the inner group structure. We further assume the
auctioneer is a trusted party who will follow the mechanism as
specified.\footnote{In the context of DAOs, the lower-level mechanism and aggregation functions can be run programmatically on a smart contract, eliminating the need to actually appoint a trusted auctioneer.}

Each item $l$ is constrained to being allocated to a single group but
there are no constraints on how many bidders within that group can have
access to the item. In other words, given that the auctioneer
allocates $l$ to $G^j$, $G^j$ has full autonomy on deciding what
subset of its members get allocated (\textit{i.e.}, granted access to) item $l$. In this sense, a group treats
an allocated item as a public excludable good.

We will refer to the $i$th bidder in group $j$ as bidder $i^j$. Each
bidder has a valuation function for the items they receive
$v_i^j:\mathcal{P}([m])\rightarrow \realspos$. We assume bidders have
quasi-linear utilities where if bidder $i^j$ is allocated a set of items $S_i^j$ and pays $p_i^j$ then $u_i^j(S_i^j,p_i^j)=v_i^j(S_i^j)-p_i^j$. Let $\biddingLang$ denote the bidding language for bidders. $\biddingLang$ will always be expressive enough for bidders to express their true valuation function. Each bidder acts strategically to submit a bid $b_i^j\in \biddingLang$ to their group seeking to maximize their utility.

We consider two main classes of valuation functions for bidders in this work:

\subparagraph*{Additive} Additive bidders have a value for each item and their value for a set of items is the sum of their values for the individual items in that set. 
Formally, for each of the items $l\in[m]$ each bidder $i^j$ has some value $v_i^j(l)\in \realspos$. Then bidder $i^j$'s value for a set $S\subseteq [m]$ of items is $v_i^j(S)=\sum_{l\in S} v_i^j(l)$.

\subparagraph*{Unit-Demand} Unit-demand bidders will have a value for each
item but their value for a set of items will only be the highest value
they have for any item in that set. Formally, for each of the items
$l\in[m]$ each bidder $i^j$ has some value $v_i^j(l)\in
\realspos$. Then $i^j$'s value for a set $S\subseteq [m]$ of items is
$v_i^j(S)=\max_{l\in S} v_i^j(l)$.

In both cases of valuation functions, the bidding language
$\biddingLang$ consists of vectors  $b_i^j\in \realspos^m$ where
$b_i^j(l)$ specifies bidder $i^j$'s bid for item $l$. 

We define a two-level mechanism as consisting of two parts, a lower
and upper mechanism. The upper mechanism is run by the auctioneer and
takes as input bids from each of the groups. The auctioneer then
decides which items should be allocated to which groups and how much
those groups should pay. The upper mechanism falls into the typical
mechanism design framework. The lower mechanism is run by each group
and dictates how the group should aggregate bids from its members into
a group bid. Then, given an allocation of items and a payment request
from the auctioneer, the lower mechanism specifies how a group should
assign items to bidders in the group and how much each bidder should
pay to cover the group's payment to the auctioneer. In this work we will
only consider deterministic mechanisms of this form. 

Formally, a lower mechanism $\lowermech=(a,\loweral,c)$ for a group with $n$ bidders consists of:
\begin{itemize}
\item An \emph{aggregation rule} $a:\biddingLang^n\rightarrow \biddingLang$
  mapping a vector of member bids into a single bid for the
  group. 
We insist that $a$ is the identity function when~$n=1$; Intuitively, if the upper mechanism and lower mechanisms are truthful, there is no reason for the aggregation function to distort the bid of a bidder in a group of size one.
\item A (lower) \emph{allocation rule}
  $\loweral:\biddingLang^n\times \mathcal{P}([m])\times \realspos
  \rightarrow A^l$.
  Each element of $A^l$ specifies which items each bidder has
  access to based on which items the group is allocated and how much
  the group has to pay. We will sometimes refer to a specific lower
  allocation by $\loweral=(x^l_1,...,x^l_{n_j})$ where each $x^l_i$
  returns which set of items $i^j$ is allocated.

\item A \emph{cost-sharing rule}
  $c:\biddingLang^n\times \mathcal{P}([m])\times \realspos \rightarrow
  \realspos^n$
  that specifies how much each bidder in the group has to pay based of
  which items they are given access to and how much the overall group
  has to pay. $c_i^j$ will be the function that specifically denotes
  how much bidder $i^j$ has to pay.
\end{itemize}

An upper mechanism $\uppermech=(\upperal,p)$ with $k$ distinct groups consists of:
\begin{itemize}
    \item An (upper) \emph{allocation rule} $\upperal:\biddingLang^k\rightarrow A^u$ where each element of $A^u$ specifies which items each group is allocated. We constrain the allocation rule such that the mechanism can allocate each item to at most 1 group. We will sometimes refer to a specific upper allocation by $\upperal=(x^u_1,...,x^u_k)$ where each $x^u_j$ returns which set of items group $G^j$ is allocated.  
    \item A \emph{payment rule} $p:\biddingLang^k\rightarrow \realspos^k$ specifying how much the upper mechanism charges each group. $p^j$ will be the function that specifically denotes how much group $G^j$ has to pay. We sometimes abuse notation and use $p^j$ to also refer to the specific amount group $j$ has to pay in a particular instance.
\end{itemize}

We can now formally define a two-level  mechanism. 

\begin{definition}
A two-level mechanism $\mech=(\uppermech,\lowermech)$ is defined by a pair of lower and upper mechanisms.
\end{definition}

An allocation of items to bidders is $\{S_i^j\}_{i,j}$. This implies an allocation of $S^j= \cup_{i=1}^{n_j} S_i^j$ to each group. We say an allocation is feasible if $S^{j_1}\cap S^{j_2}=\emptyset$ for all $j_1,j_2\in [k]$ where $j_1\neq j_2$. This implies that an allocation is feasible as long as no item is allocated to more than one group.

The social welfare of an allocation is given by
$\sum_{j=1}^k\sum_{i=1}^{n_j}v_i^j(S_i^j)$. The optimal social welfare
for an instance $I$, \OPT$(I)$  is the maximum social welfare
obtainable over feasible allocations given the valuation functions specified by $I$.\footnote{In any optimal
  allocation, every group might as well give all of its allocated
  items to all of its members.} We refer to the social welfare a
mechanism $\mech$ achieves in some instance $I$ by $SW(\mech(I))$. We say that a mechanism achieves an $\alpha$ approximation to
optimal welfare if for every instances $I$, we have that $\alpha
\ge \frac{\OPT(I)}{SW(\mech(I))}$. 

Given this setup we seek the following properties from any two-level
mechanism: 

\subparagraph*{Incentive-Compatibility} A mechanism $\mech$ is incentive
compatible if for all instances, each bidder has a dominant strategy
to report their true valuation as their bid. More formally, if $v_{-i}^j$ denote the
values of all bidders apart from $i^j$, then $u_i^j(v_i^j,v_{-i}^j)\ge
u_i^j(\tilde{v}_i^{j},v_{-i}^j)$ for any $\tilde{v}_i^j\in V,
v_{-i}^j\in V^{n-1}$ where $V$ is the set of possible valuations for any bidder. 

\subparagraph*{Budget-Balance} A mechanism $\mech$ is budget balanced if for every group, the payment by a group's members exactly covers the cost charged by the auctioneer, $\sum_{i=1}^{n_j} p_i^j = p^j \ \forall j\in[k]$.

\subparagraph*{Individual Rationality} A mechanism $\mech$ is individually
rational if bidders that bid truthfully always have non-negative
utility. That is, for any bidder $i^j$ that bids truthfully, if $\mech$
outputs an allocation $\{S_i^j\}_{i,j}$ and prices $\{p_i^j\}_{i,j}$,
then $v_i^j(S_i^j)\ge p_i^j$.\\

The following properties (of a two-level mechanism) will also be important for some of our lower
bound proofs:

\subparagraph*{Equal Treatment}
A mechanism $\mech$ satisfies equal treatment if any two bidders in
the same group $i_1^j,i_2^j$ that make the same bids also receive the
same allocation and the same payment. That is, if
$b_{i_1}^j=b_{i_2}^j$, then the allocation $\{S_i^j\}_{i,j}$
and prices $\{p_i^j\}_{i,j}$ output by $\mech$ must satisfy
$S_{i_1}^j=S_{i_2}^j$ and $p_{i_1}^j=p_{i_2}^j$.

\subparagraph*{Consumer Sovereignty} A mechanism $\mech$ satisfies
consumer sovereignty if a bidder can force the mechanism to allocate
it a specific bundle by bidding sufficiently high.  Formally,
for any bidder $i^j$, given bids by all other
bidders $b_{-i}$, there exists some bid $b_i^j$ such that $\mech$
outputs an allocation $\{S_i^j\}_{i,j}$ where $v_i^j(S_i^j)=\max_{S\subseteq [m]}v_i^j(S)$.

\subparagraph*{Upper Semi-Continuity} A mechanism $\mech$ satisfies upper
semi-continuity if for any bidder $i^j$, given bids $b_{-i}$ by the
outputs an allocation where $v_i^j(S_i^j)=\max_{S\subseteq [m]}v_i^j(S)$
other bidders, if $i^j$ is allocated $S_i^j$ for all bids $\tilde{b}_i^j \not= b_i^j$ with
$\tilde{b}_i^j \ge b_i^j$ (component-wise over items), then $i$ is
allocated $S_i^j$ by bidding $b_i^j$ as well. 

\section{Single-Item Mechanisms}
We begin our investigation of two-level mechanisms in the canonical
setting of single-item auctions.
The item can be allocated to one group and that group can 
grant access to the item to any
subset of its members. Even in this
simple setting, we prove that no incentive-compatible mechanism can
achieve a constant-factor approximation of the optimal social
welfare. Instead, we provide a mechanism that achieves a
$H_n\approx \ln n$ approximation and show that this is the best any
truthful mechanism can do. 

Denote the single item by $g$. Then every bidder $i^j$ has a value
$v_i^j\in \realspos$ if they are allocated the item and value 0
otherwise. The bidding language is $\biddingLang=\realspos$.

\subsection{Truthful Mechanism}

We propose a two-level mechanism that is truthful, budget-balanced,
individually rational, and satisfies equal treatment while obtaining a
$H_{\ell}$-approximation of the optimal social welfare (where~$\ell$
denotes the largest number of members of any group and
$H_{\ell} = \sum_{i=1}^{\ell} 1/i$ the $\ell$th Harmonic number). 

Since $b_i^j\in \realspos$, assume without loss of generality that for
each group $G^j$, $b_1^j\ge b_2^j \ge \ldots \ge b_{n_j}^j$. The upper
allocation is represented by a vector $x^u\in \{0,1\}^k$ where
$x^u_j=1$ if group $G^j$ is allocated the item by $\uppermech$ and
$x^u_j=0$ otherwise. Similarly, the lower allocation for group $G^j$
is a vector $x^l\in\{0,1\}^{n_j}$ where $x^l_i=1$ if bidder $i^j$ is
allocated the item and $x^l_i=0$ otherwise.

The lower mechanism aggregates bids by calculating each group's
\emph{willingness to pay}. We define a group $G^j$'s
willingness to pay $\text{WTP}^j=\text{WTP}(b_1^j,\ldots,b_{n_j}^j)$ as the maximum amount the group can pay assuming
that everyone who gets the item will pay the same amount (and no bidder pays more than its bid).\footnote{We will see in Lemma~\ref{l:lb} that this equal payment condition is necessary for incentive-compatibility (modulo some degenerate cases).} That is, the aggregation rule is given by

\begin{equation*}
    a(b_1^j,\ldots,b_{n_j}^j)=\text{WTP}(b_1^j,\ldots,b_{n_j}^j)=\max_{i=1,\ldots,n_j}\{ib_i^j\}.
\end{equation*}

Let $t^j(p)= \max_{i=1,\ldots,n_{j}}\{i \mid ib_i^{j} \ge p
\}$. In words, $t^j(p)$ is the largest number of bidders in group $j$ that could be allocated the item and pay equally for it (without any bidder paying larger than its bid) \textit{assuming the group is charged $p$ by the auctioneer}. Note that, if $p\le \text{WTP}^j$, then $t^j(p)$ is well defined. If group $j$ is
allocated the item by the upper mechanism and charged $p^j$ by the auctioneer, we
define the lower allocation and cost sharing rules as follows for all
$i=1,\ldots,n_j$:

\begin{align*}
    &\loweral_i(b_1^j,\ldots,b_{n_j}^j,S^j,p^j) = \begin{cases}
   1 & \text{ if } i\le t^j(p^j)\\
   0 & \text{ else }
   \end{cases} \\
   & c_i(b_1^j,\ldots,b_{n_j}^j,S^j,p^j) = \begin{cases}
   \frac{p^j}{t^j(p^j)} & \text{ if } i\le t^j(p^j)\\
   0 & \text{ else }
   \end{cases}
\end{align*}

Otherwise, if group $j$ is not allocated the item, we simply have
$\loweral_i =0$ and $c_i=0$ for all $i\in[n_j]$. 

The upper mechanism is a Vickrey auction. It takes all the group bids
and allocates the item to the group with the highest bid and charges
them the second highest group's bid. Formalizing this, given group
bids $b^1,\ldots,b^k$ and letting $j^*$ denote the index of the
highest-bidding group:
\begin{equation*}
    \upperal_j(b^1,\ldots,b^k) = \begin{cases}
    1 & \text{ if } j=j^* \\ 
    0 & \text{ else }
    \end{cases} \ \ \ \ 
    p^j(b^1,\ldots,b^k) = \begin{cases}
    \max_{j\neq j^*}\{b^j\} & \text{ if } j=j^* \\ 
    0 & \text{ else }
    \end{cases}
\end{equation*}
with ties broken arbitrarily.

Informally our mechanism works by having each group calculate their willingness to pay and bid that amount. Then the upper mechanism  runs a standard second price auction. The winning group then splits the cost it is charged equally amongst the largest subset of its agents that it is able to. This subset of agents that are able to equally split the cost are exactly the agents the group gives access of the item to. Formally, the lower and upper mechanisms can be implemented together as follows. For the ease of exposition, any bidders whose
allocations/payments aren't explicitly listed are implied to not be allocated any items and to have zero payment.

\begin{algorithm}[H]\label{mech1}
\SetKwInOut{Input}{Input}
\DontPrintSemicolon
\Input{Bids $b^j=(b_1^j,\ldots,b_{n_j}^j)$ with $b_{i}^j\ge b_{i+1}^j$ for $j\in[k]$, $i\in[n_j-1]$}
\For{j=1,\ldots,k}{
    $\text{WTP}^j\gets \max_{i\in[n_j]}\{ib_i^j\}$\;
}
$j^*\gets \argmax_{j\in[k]}\{\text{WTP}^j\}$\; 
$p^{j^*}\gets \max_{j\neq j^*}\{\text{WTP}^j\}$\;
$i^* \gets \max_{i\in[n_{j^*}]}\{i|ib_i^{j^*} \ge p^{j^*} \}$\;
\For{$i=1,\ldots,i^*$}{
    $x_i^{j^*}\gets 1$\; 
    $p_i^{j^*}\gets \frac{p^{j^*}}{i^*}$\; 
}
\Return{Allocation $\vec{x}$, payments $\vec{p}$ } 
 \caption{ Single-Item Two-level  Mechanism}
\end{algorithm}

\begin{theorem}
Mechanism~\ref{mech1} is truthful, budget-balanced, and individually rational. 
\end{theorem}

\begin{proof}
We first show that the mechanism is truthful. Since this is a single
parameter setting, it suffices to show that the allocation rule is monotone
and that each winning bidder pays their critical bid.

We first show that the allocation rule is monotone. Assume bidder
$i^{j^*}$ is allocated the item by bidding $b_i^{j^*}$. Then if they
increase their bid to $\tilde{b}_i^{j^*}\ge b_i^{j^*}$, we have that
$\text{WTP}^{j^*}$ weakly increases (that is, willingness-to-pay is weakly monotone in any coordinate).
Thus, if group $G^{j^*}$ wins when bidder
$i^{j^*}$ bids $b_i^{j^*}$, then group $G^{j^*}$ will still win when bidder
$i^{j^*}$ bids $\tilde{b}_i^{j^*}$ with all other bidders' bids kept
fixed. Furthermore, $p^{j^*}$ would stay constant since the other
bidders' bids stayed constant, and thus
$b_i^{j^*}\ge \frac{p^{j^*}}{i^*}$ would imply
$\tilde{b}_i^{j^*}\ge \frac{p^{j^*}}{i^*}$. Therefore, bidder $i^{j^*}$ would
still be allocated the item after increasing their bid, showing the
allocation rule is monotone.

We now show that each winning bidder pays their critical bid, by showing that $\frac{p^{j^*}}{i^*}$ is the lowest that bidder
$i^{j^*}$ could drop their bid to while still winning (that is, bidder $i^{j^*}$'s payment is equal to its critical bid). Assume that bidder
$i^{j^*}$ drops their bid from $b_i^{j^*}$ to
$\tilde{b}_i^{j^*}=\frac{p^{j^*}}{i^*}$. By the definition of $i^*$
there are still at least $i^*-1$ other bidders in $G^{j^*}$ with bids
at least $\frac{p^{j^*}}{i^*}$. Thus $G^{j^*}$'s WTP remains at least
$p^{j^*}$ implying that $G^{j^*}$ still wins the upper
auction. Furthermore, bidder $i^{j*}$ would still have at least the
$i^*$th highest bid in $G^{j^*}$. Thus we would still have
$i^*\tilde{b}_i^{j^*}\ge p^{j^*}$ implying that bidder $i^{j^*}$ would
remain part of the winning subset of $G^{j^*}$.

If bidder $i^{j^*}$ drops their bid to $b_i^{j^*}$ below $\frac{p^{j^*}}{i^*}$ then bidder $i^{j^*}$ becomes at best the $i^*$th highest bidder in $G^{j^*}$. Even if $G^{j^*}$ is still the winning group, $G^{j^*}$ still has to pay the same payment $p^{j^*}$. This implies that $\max_{i=1,\ldots,n_{j^*}}\{i|ib_i^{j^*} \ge p^{j^*} \}$ weakly decreases by bidder $i^{j^*}$ decreasing their bid. Thus if bidder $i^{j^*}$ falls below the $i^*$th bid in $G^j$ then they will no longer be part of the winning set. And, even if bidder $i^{j^*}$ does become the $i^*$th highest bidder in $G^j$, then we would have $i^*\tilde{b}_i^{j^*}< i^*\frac{p^{j^*}}{i^*}=p^{j^*}$ implying that bidder $i^{j^*}$ would no longer be part of the winning set. Thus bidder $i^{j^*}$ would never be part of the winning set by dropping their bid below $\frac{p^{j^*}}{i^*}$. This shows that $p_i^{j^*}=\frac{p^{j^*}}{i^*}$ is bidder $i^{j^*}$'s critical bid for all winning bidders $i^{j^*}$. 

Budget-balance is trivial for every losing group since every losing
group is charged 0 by the auctioneer and doesn't have any of its
members make payments. When the winning group is charged $p^j$, it
chooses $i^*$ bidders to pay $\frac{p^j}{i^*}$, showing that
budget-balance holds there as well. Individual rationality follows
since the only bidders that make payments are chosen such that
(assuming truthful bids)
$v_i^j\ge \frac{p^j}{i^*}$ and so $u_i^j=v_i^j-\frac{p^j}{i^*}\ge 0$.
\end{proof}

\begin{theorem}\label{t:hl}
Assuming truthful bidding, Mechanism~\ref{mech1} achieves an
$H_{\ell}$-approximation to the optimal social welfare, where~$\ell$
is the maximum size of a group.
\end{theorem}

\begin{proof}
We start with the following lemma giving a lower bound for the WTP of a group compared to the total value that group would get if every member was allocated the item. 

\begin{lemma}\label{lem:harmonic}
$\text{WTP}^j \ge \frac{W^j}{H_{n_j}}$ where $W^j=\sum_{i=1}^{n_j} v_i^j$ and $H_i$ is the $i$th harmonic number.  
\end{lemma}

\begin{proof}
Note that $\text{WTP}^j=\max_{i=1,\ldots,n_j}\{iv_i^j\}$. Thus we have,
\begin{equation*}
     \text{WTP}^j \cdot H_{n_j}=\max_{i=1,\ldots,n_j}\{iv_i^j\}\sum_{i=1}^{n_j} \frac{1}{i} \ge \sum_{i=1}^{n_j} \frac{1}{i}iv_i^j = W^j.
\end{equation*}
\end{proof}

(Returning to the proof of Theorem~\ref{t:hl}.) Note that if $G^j$ is the group that wins the upper
mechanism, then they obtain value at least $\text{WTP}^j$ amongst their group
members by allocating the item within their group. This is because if
$G^j$ wins, we have $\text{WTP}^j\ge p^j$ and so
$i^*= \max\{i|iv_i^j\ge p_j\}\ge
\argmax_{i=1,\ldots,n_j}\{iv_i^j\}$.
Thus, $\text{WTP}^j$ is a lower bound of the value of the $i^*$ bidders with
the highest values in $G^j$ being allocated the item.

Note that the optimal social welfare is achieved by the the group with
the highest $W^j$ to receive the item and for every member of that
group to be allocated that item. Assume WLOG that this is $G^1$ and
some group $G^j$ wins the upper auction. Then we have that the
mechanism achieves welfare at least $\text{WTP}^j$. If $G^j=G^1$ then we are
done by Lemma~\ref{lem:harmonic}; otherwise the fact that $G^j$ won over $G^1$ in the
upper auction implies $\text{WTP}^j\ge \text{WTP}^1\ge \frac{W^1}{H_{n_1}}$. Since
$H_{n_j} \le  H_{\ell}$, we have that the mechanism always achieves a $H_{\ell}$ fraction of the optimal social welfare. 
\end{proof}

In general, the approximation factor achieved by Mechanism~\ref{mech1} is
governed by the ratio of $W^j$ and $\text{WTP}^j$.  If this ratio is known to
be smaller than~$H_{\ell}$---for example, because members of a common
group tend to have similar valuations---then the guarantee of
Theorem~\ref{t:hl} improves accordingly.

\subsection{Lower Bounds}

We now show that the $H_{\ell}$-approximation to welfare achieved by Mechanism~\ref{mech1} is in fact
the best we can hope for from any truthful, budget-balanced mechanism
satisfying equal treatment.  We begin by bounding the maximum amount a
group can be induced to pay compared to their true value for an item
while maintaining incentive compatibility. Then we give a specific
instance in which this occurs and show that any truthful budget
balanced mechanism necessarily has to sometimes allocate items to
lower-valued groups.

\begin{theorem}\label{t:lb}
There is no truthful, budget-balanced, individually rational two-level
mechanism that satisfies equal treatment and guarantees more than an
$H_{\ell}$ fraction of the optimal social welfare (where~$\ell$ is the
maximum group size).
\end{theorem}
\begin{proof}
First, we can restrict attention to mechanisms that satisfy consumer
sovereignty.
If a mechanism doesn't satisfy consumer sovereignty,
there exists an instance in which there is a bidder in some group
that will never be allocated the item regardless of what they bid. Since
the threshold price that a bidder needs to pay to be allocated the item is
not a function of their bid for truthful mechanisms, this must hold
regardless of the bidder's valuation.  Letting that bidder's valuation
tend to infinity then shows that worst-case welfare approximation of
the mechanism is arbitrarily bad.

Next,
note that because of the restriction that the aggregation function
must be the identity function in the case in which a group only has 1
bidder (see Section~\ref{sec:model}), for the mechanism to be
truthful, we must have that the upper mechanism is truthful with
respect to group bids to be truthful with respect to individual bidders. Given this, we show the following result
constraining the cost-sharing rule within the winning group.  

\begin{lemma}\label{l:lb}
In single-item settings,
except possibly on a set of valuation profiles with Lebesgue measure
zero, a truthful and budget-balanced mechanism that satisfies equal
treatment and consumer sovereignty
must always charge all winning bidders the same payment.

\end{lemma}
\begin{proof}
Note that winning bidders must all come from the same group. Then,
because the upper mechanism is a truthful single-item auction,
conditioned on $G^j$ winning the item, the payment $p^j$ that $G^j$
has to make is independent of their bid $b^j$ and hence
$(b_1^j,\ldots,b_{n_j}^j)$. Thus in a truthful, budget-balanced
mechanism, $G^j$ wins if and only if $\lowermech$ is such that the
payments its cost-sharing rule charges to winning bidders can cover
the group's threshold payment $t^j(b^{-j})$. This makes $\lowermech$ a
truthful, budget-balanced cost sharing mechanism where $\mech$ being
truthful and budget-balanced implies that $\lowermech$ will always be
able to cover the cost it is asked to. 

Next, we invoke a characterization result from~\cite[Theorem
3.4]{shapley} that implies that, except possibly on a set of
valuation profiles with Lebesgue measure zero, $\lowermech$ must be
the same lower mechanism as in Mechanism~\ref{mech1}, meaning that it
identifies the maximal subset of bidders $S\subset G^j$ such that each bidder
$i^j\in S$ has  $b_i^j\ge \frac{p^j}{|S|}$ and charges them all
$\frac{p^j}{|S|}$. (Note that~$S$ is uniquely defined, as the union of
two sets satisfying this property also satisfies that 
property.)  In particular, every winning bidder makes the same payment. 
\end{proof}

(Returning to the proof of Theorem~\ref{t:lb}).  
Consider now an incentive-compatible two-level mechanism in which the
lower mechanisms satisfy equal treatment.
Recall that $\text{WTP}^j$ as the largest amount that $G^j$ can pay assuming all the winning bidders pay the same amount.
By Lemma~\ref{l:lb}, the lower mechanisms must charge winning bidders a common amount (subject to
individual rationality), so $\text{WTP}^j$ is the maximum payment that could possibly be made by group $G^j$ (except possibly on a measure-zero set of valuation profiles).

With this fact in hand, we consider two different instances, both
with two groups.  We assume that these instances do not fall into
the measure-zero set of valuation profiles mentioned above; this
assumption can always be enforced through arbitrarily small
perturbations to the valuations, if needed.\\
\noindent
\textbf{Instance 1:} $G^1$ has $n-1$ bidders with $v_i^1=\frac{1}{i}-\delta$ for some $\delta>0$ and $G^2$ has one bidder with $v_1^2=1$.\\
\noindent
\textbf{Instance 2:}  $G^1$ has one bidder with $v_1^1=1$ and $G^2$ has $n-1$ bidders with $v_i^2=\frac{1}{i}-\delta$ for some $\delta>0$.\\
\noindent
We refer to Instance 1 by $I_1$ and Instance 2 by $I_2$. Note that the
optimal welfare in both instances is $H_{n-1}$ where all the bidders
in $G^1$ win in $I_1$ and all the bidders in $G^2$ win in $I_2$. If
$G^2$ wins in $I_1$ or $G^1$ wins in $I_2$ then the mechanism will
only achieve a welfare of 1. Assume there is a truthful,
budget-balanced two-level mechanism $\mech$ that achieves an
approximation factor better than $H_{n-1}$.  Then, assuming that
$\delta$ is sufficiently close to~0, $\mech$ must have $\uppermech$
choose $G^1$ in $I_1$ and $G^2$ in $I_2$.

Since $\uppermech$ must be truthful with respect to group bids, it can be characterized by threshold payments $t^1(b^2)$ and $t^2(b^1)$ where $G^1$ winning implies $b^1\ge t^1(b^2) $ and $b^2 \le t^2(b^1)$ and $G^2$ winning implies the opposite. Since $a$ is the identity function for groups with a single member, we have in $I_1$ that $b^2=1$ and in $I_2$ that $b^1=1$. We claim that $G^1$ winning in $I_1$ and $G_2$ winning in $I_2$ imply $\max\{t^1(1),t^2(1)\}\ge 1$.

Assume otherwise, and let $\max\{t^1(1),t^2(1)\}=y$ with $y<1$. Then letting $\epsilon >0$ such that $y+\epsilon <1$, it follows that if $b^1=y+\epsilon$ in $I_1$ and $b^2=y+\epsilon$ in $I_2$ then we would still have $G^1$ win in $I_1$ and $G^2$ win in $I_2$. However this would imply $t^1(y+\epsilon)\ge 1 \implies t^1(1) < t^1(y+\epsilon)$ contradicting the monotonicity and hence truthfulness of $\uppermech$. 

Hence $\max\{t^1(1),t^2(1)\}\ge 1$ implies that either $G^1$ pays at
least 1 in $I_1$ or $G^2$ pays at least 1 in $I_2$. However in $I_1$,
we have $\text{WTP}(G^1)=1-\delta$ and in $I_2$, we have
$\text{WTP}(G^2)=1-\delta$. Thus by the above lemma, there is no truthful
budget-balanced mechanism satisfying equal treatment that will choose
both $G^1$ to win in $I_1$ and $G^2$ to win in $I_2$ implying that no
truthful, budget-balanced mechanism satisfying equal treatment can do
better than a $H_{n-1}$ approximation factor.
\end{proof}

This lower bound shows that the $H_{\ell}$-approximation of Mechanism~\ref{mech1} is in fact tight and the best any incentive-compatible two-level
mechanism could hope to do in this setting.

\section{Multi-Unit Mechanisms}
In this section, we move beyond the single-item setting and consider
the case where the auctioneer is selling multiple items. Bidders can
now have combinatorial valuations over the different
items. In the case where bidders have additive valuations across
the different items we show that our positive result for the single-item
settings (Theorem~\ref{t:hl})
extends naturally. 

The main result of this section shows that in general---and even with
unit-demand valuations---
we can't hope for any non-trivial approximation 
to optimal social welfare. In particular, we show that no truthful
budget-balanced mechanism can do better than an $n$-approximation to
social welfare in this setting, under very weak assumptions on the
aggregation function.

\subsection{Additive Valuations}

For each of the items $l\in[m]$, each bidder $i^j$ has some value
$v_i^j(l)\in \realspos$. Then $i^j$'s value for a set $S\subset [m]$
of items is $v_i^j(S)=\sum_{l\in S} v_i^j(l)$. The bidding language
$\biddingLang$ consists of vectors  $b_i^j\in \realspos^m$ where
$b_i^j(l)$ specifies bidder $i^j$'s value for item $l$. We will use
$b^j(l)$ to refer to the vector of values all the bidders in $G^j$
have for item $l$.

To extend our upper bound for single-item settings
(Theorem~\ref{t:lb}) to the setting of additive valuations, it is
sufficient to run an independent copy of Mechanism~\ref{mech1} for each of
the items:

\begin{algorithm}[H]\label{mech2}
\SetKwInOut{Input}{Input}
\DontPrintSemicolon
\Input{Bids $b^j=(b_1^j,...,b_{n_j}^j)$ for $j=1,...,k$;\\ Single-Item
  Mechanism $\mech$ (Mechanism~\ref{mech1})}
\For{l=1,...,m}{
    $(\tilde{x},\tilde{p})\gets \mech(b^1(l),...,b^k(l))$\;
    \For{$j=1,...,k$}{
        \For{$i=1,...,n_j$}{
            $p_i^j \gets p_i^j + \tilde{p}_i^j$\;
            $x_i^j(l) \gets \tilde{x}_i^j$\;
        }    
    }
}

\Return{Allocation $\vec{x}$, payments $\vec{p}$ } 
 \caption{ Additive Item Two-level  Mechanism}
\end{algorithm}

\begin{theorem}
Mechanism~\ref{mech2} is truthful, budget-balanced, individually rational and achieves a $H_{\ell}$ approximation to the optimal social welfare for bidders with additive valuations. 
\end{theorem}

\begin{proof}
Because a bidder's bid on one item has no affect on the
allocation of or payments for any other item, truthfulness of
Mechanism~\ref{mech2} follows straightforwardly from the truthfulness of
Mechanism~\ref{mech1}.

Similarly, the budget-balance and individual rationality properties of
Mechanism~\ref{mech2} follow easily from those of Mechanism~\ref{mech1} (and, in fact,
hold on an item-by-item basis).

In the welfare-maximizing allocation, each item goes to the group that
has the highest total value for that item and has all of its members
allocated that item. From the analysis of Mechanism~\ref{mech1}, we have that it
allocates a given item to a set of bidders that have total value at
least a 1/$H_{\ell}$ fraction of the value any group has for that
item. Because bidders have valuations that are additive across items,
this implies that the welfare achieved by Mechanism~\ref{mech2} is within an
$H_{\ell}$ factor of the optimal welfare.
\end{proof}

\subsection{Unit-demand Valuations}
We now consider bidders with unit-demand valuations.
For each of the items $l\in[m]$, each bidder $i^j$ has some value
$v_i^j(l)\in \realspos$. Then $i^j$'s value for a set $S\subset [m]$
of items is $v_i^j(S)=\max_{l\in S} v_i^j(l)$. The bidding language
$\biddingLang$ consists of vectors  $b_i^j\in \realspos^m$ where $i^j$
specifies their bid for each item. 
In line with the auctioneer (of the upper mechanism) choosing a
mechanism that is agnostic to group-specific idiosyncrasies, 
we require some kind of assumption that precludes a group from using
its bid to signal information about its inner structure as opposed
to a reasonable aggregation of its members' preferences.  
Many different such assumptions would be sufficient for our purposes;
for concreteness, assume from now that the output of an aggregation
function on a specific item must be bounded by a fixed but arbitrary
function of the its inputs for that item.  That is,

there must exist some function
$f:\realspos^n\rightarrow \realspos$ such that for any item $l$ and
group $G^j$, $a(b_1^j,...,b_{n_j}^j)(l)\le
f(b_1^j(l),...,b_{n_j}^j(l))$.  

Within this setting we show that no truthful two-level mechanism can
do better than a $n$ approximation to the optimal welfare. The issue
mechanisms in this setting face is not being able to distinguish whether
a group bid representing a high value for many different items comes
from that group having many bidders with disparate preferences or from
a single unit demand bidder who is agnostic to which item they
receive. In the interest of maintaining truthfulness, the mechanism is
often forced to assume the group is composed of multiple disparate
members and allocate to that group all the items they have high value
for. However to remain truthful in the case where this valuation
actually came from a single bidder with high value, the auctioneer
can't charge more for a large bundle than it would charge for its
individual components. In cases where different groups have similar
preferences over the same items, this can cause only one of the groups
being allocated the entire set of items hence harming the welfare in
the case where these actually were just individual unit-demand
bidders.  We now proceed to making these ideas precise.

\begin{theorem}
No truthful two-level mechanism can achieve better than a $n$ fraction of the optimal welfare.
\end{theorem}

\begin{proof}
We can assume any mechanism that achieves at least a $n$-approximation to the optimal welfare satisfies consumer sovereignty. Since, if $\mech$ does not satisfy consumer sovereignty, there must exist an instance where there is some item $l$ that bidder $i^j$ can't receive regardless of how high they bid. In this case let $v_i^j(l)\rightarrow \infty$. Then $\mech$ will have an arbitrarily bad approximation factor to optimal welfare.

Throughout this proof we will use $A_u^j$ to refer to the set of possible sets of items the upper mechanism can allocate to $G^j$. 

We start by showing that any truthful mechanism satisfying consumer sovereignty must have an upper mechanism that is able to allocate the set of all items to any group. 

\begin{lemma}\label{lem:[m]}
Any truthful two-level mechanism that satisfies consumer sovereignty must satisfy $[m]\in A_u^j$ for all $j$.
\end{lemma}

\begin{proof}
Let $G^j$ have $m$ unit-demand bidders, with bidder $i^j$'s valuation function $v_i^j(l)=x$ if $i=l$ and otherwise $v_i^j(l)=0$ where $x$ is an arbitrary constant. Now fix the bids by the other groups. Since $\uppermech$ must be truthful with respect to group bids, conditioned on $\uppermech$ allocating a set $S\in A_u^j$ of items to $G^j$, $G^j$ pays the same amount $p^j(S)$ regardless of its bid $b^j$. Using this along with consumer sovereignty, we claim there exists a value of $x$ such that $\uppermech$ must allocate $[m]$ to $G^j$ to be truthful. 

Assume that $\uppermech$ does not allocate $[m]$ to $G^j$. It follows that there is some item $l$ not allocated to $G^j$ which further implies that $l^j$ will have 0 utility in this allocation. However by consumer sovereignty, there exists a bid $b_l^j$ that bidder $l^j$ can make such that $\uppermech$ will allocate some set $S$ to $G^j$ where $l\in S$ and furthermore $l^j$ will be allocated $l$ by $G^j$ in $\lowermech$. As we noted before the price $\uppermech$ can charge to $G^j$ for $S$ is some fixed price $p^j(S)$, not a function of $b^j$. Thus consider the case where $x=\max_{S\in A_u^j}\{p^j(S)\}+1$. By budget-balance, $\lowermech$ can charge bidder $l^j$ at most $p^j(S)$ given that $G^j$ is charged $p^j(S)$ by $\uppermech$. Thus bidder $l^j$ can misreport their bid to be allocated $l$ by $\lowermech$, and as a result get a utility at least $\max_{S\in A_u^j}\{p^j(S)\}+1-p^j(S) > 0$. This shows a profitable non-truthful deviation for bidder $l^j$. This implies that when $x=\max_{S\in A_u^j}\{p^j(S)\}+1$, $G^j$ must be allocated $[m]$. Since we arbitrarily fixed the bids by the other groups, it follows that we must always have $[m]\in A_u^j$ for all $j$ when $\mech$ is truthful and satisfies consumer sovereignty. 
\end{proof}

Lemma~\ref{lem:[m]} shows that for a truthful mechanism $\mech$ there must always be some bid $b^j$ that $G^j$ can make to be allocated all of $[m]$. We build upon this to show that in fact, for $\mech$ to be truthful, for any arbitrary choice of bids $b^{-j}$, for every item $l$ there exists some bid $b^j$ where $l$ is the highest value item in $b^j$ and $G^j$ is still allocated all of $[m]$. This shows that in such an instance, if $G^j$ is actually comprised of a single bidder with valuation function $b^j$, they will still be allocated all of $[m]$ even though they are perfectly happy just receiving $l$. This will be used to cap the amount the mechanism can charge a group for the entirety of $[m]$ compared to any single item to remain truthful.  

\begin{lemma}
For any truthful two-level mechanism $\mech$ that achieves at least a $n$-approximation to welfare, for all items $l=1,...,m$, there always exists a $b^j$ such that $x^u_j(b^1,...,b^k)=[m]$ and $l=\argmax_{l\in[m]}b^j(l)$.
\end{lemma}

\begin{proof}
Following closely to the previous lemma, consider the following class of groups, parameterized by $j$. Let group $\tilde{G}^j$ be defined by having $m$ bidders where each bidder prefers distinct items but the $j$th bidder in $\tilde{G}^j$ has a much larger value for their item compared to the other bidders in the group. 
Formally, for $j=1,...,m$ let group $\tilde{G}^j$ have $m$ bidders where the unit demand valuation of bidder $i^j$ is $v_i^j(l)=y$ if $i=j=l$, $v_i^j(l)=x$ if $i\neq j$ and $i=l$, and otherwise $v_i^j(l)=0$ with $y$ to later be defined as a function of $x$. Then fix the bids by other groups and let $x=\max_{S\in A_u^j}\{p^j(S)\}+1$. Let the group bid for $\tilde{G}^j$ be $\tilde{b}^j$ and let $\vec{z}^i\in\realspos^m$ be vectors parameterized by $i=1,...,m$ where $\vec{z}^i(l)=x$ for $l=i$ and otherwise $\vec{z}^i(l)=0$. 
By our assumption that the group bid for an item is bounded by some function of the individual bidders bids for an item, we have that there exists some function $f$ such that $\tilde{b}^j(l)\le f(\vec{z}^l)$ for $l\neq j$. Now let $g(x)=\max_{l\neq j}\{f(\vec{z}^l)\}$ and set $y=n^2(g(x)+x)$. Notably this means that the $j$th bidder in $\tilde{G}^j$ has much higher value for their item than the other bidders. We claim this implies that $\tilde{b}^j(j)\ge g(x)$. 

Assume otherwise that $\tilde{b}^j(j)<g(x)$. Then define $G'$ to be a group consisting of a single unit demand bidder. Let this bidder have a valuation of $v'(l)=ng(x)$ if $l=j$ and otherwise $v'(l)=0$.

Now consider two instances both with two groups and $m$ items:\\
\noindent
\textbf{Instance 1:} $G^1$ is defined as $\tilde{G}^j$ above, and $G^2$ is defined as $G'$.\\
\noindent
\textbf{Instance 2:} $G^1$ is a single unit demand bidder with valuation $\tilde{b}^j$, and $G^2$ is defined as $G'$.
\noindent
We will refer to Instance 1 as $I_1$ and Instance 2 as $I_2$.  In $I_1$ the optimal allocation is to give all the items to $G^1$ for a welfare of $n^2(g(x)+x)+(m-1)x$, and in $I_2$ the optimal allocation is to give item $j$ to $G_2$ and everything else to $G^1$ for a welfare of at least $ng(x)$.

Thus, if $\uppermech$ allocates item $j$ to $G_1$ then in $I_2$ the welfare is less than $g(x)$ implying an approximation factor strictly greater than $\frac{ng(x)}{g(x)}=n$. If $\uppermech$ allocates item $j$ to $G_2$ then in $I_1$ the welfare is at most  $ng(x)+(m-1)x$ giving an approximation factor of $\frac{n^2(g(x)+x)+(m-1)x}{ng(x)+(m-1)x} >n$, since $n>m$. However, because the aggregation function is the identity for single bidder groups, $I_1$ and $I_2$ are indistinguishable to $\uppermech$ and $\uppermech$ being deterministic implies it must allocate item $j$ to the same group in both instances. Thus if $\tilde{b}^j(j)<g(x)$ then $\mech$ must have a worst case approximation factor worse than $n$. 

Note that $\tilde{b}^j(j)\ge g(x)$ implies that $j\in\argmax_{l\in [m]}(b^j(l))$. Furthermore from how we defined $x$ and $y$, we must have that $\uppermech$ allocates all of $[m]$ to $\tilde{G}^j$ following the proof from the previous lemma. Thus taking any fixed group bids $b^{-j}$. We can chose $G^j$ to be defined as any of $\tilde{G}^i$ for $i=1,..,m$ such that for all items $l=1,..,m$ there exists a group bid $b^j$ causing $G^j$ to be allocated all of $[m]$ while having $l= \argmax_{l'\in[m]}\{b^j(l)\}$.
\end{proof}

Note that $\uppermech$ can't distinguish between the case when $G^j$ consists of a single unit demand bidder and when $G^j$ consists of multiple unit demand bidders. We show that this restrains the payments $\uppermech$ can charge for allocating sets of items to groups by considering the case where every group is a single unit demand bidder. In characterizing the amount $\uppermech$ can charge $G^j$ for allocating $G^j$ a set $S$ of items, let the bids by all other groups be fixed. Then as before we denote the amount $\uppermech$ charges $G^j$ for $S$ as $p^j(S)$. We start by showing that any truthful $\uppermech$ can't charge less for a superset of another set of items. In the following lemmas we will refer to the unit demand valuations of the single bidder in $G^j$ by $v^j$. 

\begin{lemma}
For any truthful, two-level mechanism $\mech$, let $S,T\in A_u^j$. Then $S\subset T$ implies that $p^j(S)\le p^j(T)$.
\end{lemma}

\begin{proof}
$p^j(S)$ and $p^j(T)$ must remain constant regardless of the inner structure of $G^j$, thus without loss of generality assume $G^j$ consists of one bidder.  Then for the sake of contradiction, assume that $p^j(S)>p^j(T)$. $S\in A_u^j$ implies that there exists some valuation $v^j$ such that $G^j$ is allocated $S$ and charged $p^j(S)$. However $T\in A_u^j$ implies there is also an alternative bid $v^{j'}$ that $G^j$ could make to be allocated $T$ instead. Since $v^j$ is a unit demand valuation, we have that $v^j(T)\ge v^j(S)$. Thus $p^j(S)>p^j(T)$ would imply that $G^j$ profits by reporting $v^{j'}$ over $v^j$ making $\mech$ not truthful.  
\end{proof}

We now show the more surprising statement 
that any truthful $\uppermech$ can't charge more for $[m]$
than it would charge for any set $S$ of items. 
This
result stems from the fact that at times $\uppermech$ has to allocate
multiple items to a group that could be a unit demand bidder. Thus to
maintain truthfulness, if $G^j$'s favorite item is in $S$ then $G^j$
can't be charged more for $[m]$ than it would have been for $S$ since
it has the same value for both sets.  

\begin{lemma}
For any truthful two-level mechanism $\mech$ and set $S\in A_u^j$ where $ S\neq \emptyset$, we have $p^j(S)\ge p^j([m])$ 
\end{lemma}

\begin{proof}
As in the previous lemma, we can assume without loss of generality that $G^j$ consists of 1 bidder. We have shown that $[m]\in A_u^j$ for any truthful mechanism. Thus there exists some valuation $v^j$ such that $\uppermech$ allocates $[m]$ to $G^j$. Let $l^*=\argmax_{l\in [m]}\{v^j(l)\}$. It follows that for any set $S\in A_u^j$ such that $l^*\in S$ that $v^j(S)=v^j([m])$. $S\in A_u^j$ also implies there exists some bid $v^{j'}$ $G^j$ could make to be allocated $S$. Thus if $p^j(S)< p^j([m])$ we would have that when $G^j$'s value is $v^j$, $G^j$ has a profitable deviation by bidding $v^{j'}$ instead making $\mech$ not truthful. By lemma 10 we have that there exist valuations $v^j$ such that $\uppermech$ allocates $[m]$ to $G^j$ where $l^*=l$ for any $l\in [m]$. Thus as long as $S\neq \emptyset$, we have that $p^j(S)\ge p^j([m])$.
\end{proof}

\begin{corollary}
In any truthful mechanism $\mech$, for any 2 sets $S,T\in A_u^j$ we have $p^j(S)=p^j(T)$ $\forall j=1,...,k$. 
\end{corollary}

This follows from the fact that for any set $S\in A_u^j, p^j(S)\le p^j([m]),$ $p^j(S)\ge p^j([m])\implies p^j(S)=p^j([m])$. Finally we provide an instance that shows that no $\uppermech$ that charges the same amount for all sets can do better than a $n$ approximation of the optimal welfare.

\begin{lemma}
If $\mech$ is truthful and $\uppermech$ has $p^j(S)=p^j(T)$ $\forall S,T\in A_u^j, j=1,...,k$ then $\mech$ can not achieve a welfare approximation better than $n$.
\end{lemma}

\begin{proof}
$\uppermech$ being truthful with respect to group bids implies that the allocation $A^j$ it gives to group $G^j$ must satisfy $A^j=\argmax_{A\in A_u^j}\{v^j(A)-p^j(A)\}$. We have that $p^j(A)$ is a constant for all $A\in A_u^j$, unless $A=\emptyset$ where $p^j(A)=0$ by individual rationality, so let $p^j(A)=p^j$. However this implies that if $v^j(A)>p^j$ for any set $A$ then $A^j=\argmax_{A\in A_u^j}\{v^j(A)\}$. Let $A^{*j}\in \argmax_{A\in A_u^j}\{v^j(A)\}$. Then if $G^j$ is not allocated a set $S$ where $v^j(S)=v^j(A^{*j})$, we have that $v^j(A^{*j})\le p^j$. This implies if $v(A)<v(A^{*j})$ then $v(A)<p^j$. Thus if $\uppermech$ does not give $G^j$ a set $S$ where $v^j(S)=v^j(A^{*j})$, $\uppermech$ must give $G^j$ no items at all. 

Given this consider the following  instance with $n$ identical groups where every group consists of a single unit demand bidder. There are $n$ total items with each bidder in each group having an identical valuation vector of $v_1^j=(1+\epsilon,1,...,1,1)$. Call the first item that every bidder has value $1+\epsilon$ for $l_1$. Here we have that $v(S)=v(A^{*j})$ if and only if $l_1\in S$. However note that $\uppermech$ can only allocate $l_1$ to one of the groups. Let this group be $G^{j'}$. Then by the observation above we have that any groups $G^{j}$ where $j\neq j'$ must not be allocated any items at all. Thus the maximum welfare any truthful $\uppermech$ can achieve in this scenario is to allocate every item to $G^{j'}$ for a welfare of $1+\epsilon$. However, since every group is a single unit demand bidder, the optimal allocation is to give every group a unique item for a welfare of $n+\epsilon$. Thus as $\epsilon\rightarrow 0$ we get that no $\mech$ where $\uppermech$ charges equal prices for every set of goods can achieve better than a $n$ approximation of the optimal welfare. 
\end{proof}
Thus since every truthful upper mechanism must charge the same prices for every subset of items in its range, we have that no truthful mechanism can achieve an approximation factor better than $n$.
\end{proof}

\newpage

\begin{appendices}

\section{Truthful Mechanisms Don't Compose}\label{app:nocompose}
We give here an explicit example showing that a two-level mechanism
$\mech$ in which $\uppermech$ implements a truthful auction mechanism
and $\lowermech$ implements a truthful cost sharing mechanism doesn't
necessarily imply that $\mech$ itself is truthful two-level
mechanism. To see this, we will consider a simple instance with
unit-demand bidders and consider the two-level mechanism $\mech$
where 
$\uppermech$ implements VCG and $\lowermech$ implements the 
``maximum equal-split'' cost shares used in Mechanism~1 (see
Section~3), which are arguably the most canonical
choices for the upper and lower mechanisms.
We show that, no matter what aggregation
function is used, this two-level mechanism cannot be truthful.

Consider the following instance with 2 groups $G^1$ and $G^2$
competing for 2 items $r$ and $s$ where $G^1$ has 2 agents and $G^2$
has 1 agent. Let the agents' valuations be $v_1^1=(10+2\epsilon,0),
v_2^1=(5-\epsilon,\epsilon), v_1^2=(10,\epsilon)$ with the first
indices and second indices referring to agents' values for $r$ and $s$
respectively. Then let $\mech$ be a two-level mechanism where
$\uppermech$ implements VCG and $\lowermech$ implements maximum
equal-split cost shares.
Then since we would have $b^2=(10,\epsilon)$,
regardless of the aggregation function used to create $b^1$,
$\uppermech$ will always give at least one of the items to $G^2$ as
the welfare maximizing outcome according to group bids. Then by the
VCG payment rule, in the case that $G^1$ wins $r$ $p^1=10-\epsilon$,
and in the case that $G^2$ wins $s$, $p^1=0$.

With our
choice of cost-sharing,
all winning agents in a group pay the
same amount. Thus, since agent $2^1$'s value for $r$ is only
$5-\epsilon$, $\lowermech$ will only allocate $r$ to agent $2^1$ if
$p^1\le 10-2\epsilon$. Thus $p^1>10-2\epsilon$ implies that agent
$2^1$ is not allocated $r$ and has 0 utility in this case. However, by
consumer sovereignty, there exists some bid agent $2^1$ can make that
causes $G^1$ to be allocated item $s$ instead. Then $p^1=0$ so both
agents in $G^1$ will be allocated $s$ making $u_2^1=\epsilon >0$. Thus
no truthful mechanism can give $G^1$ only item $r$. However, if
$\uppermech$ only gives $G^1$ item $s$, then agent $1^1$ is guaranteed
to have 0 utility. Thus again by consumer sovereignty, there exists
some bid $b_1^1$ agent $1^1$ can make such that $G^1$ is allocated $r$
instead. In this case agent $1^1$ has to pay at most $10-\epsilon$ by
budget balance making $u_1^1\ge
10+2\epsilon-10+\epsilon=3\epsilon>0$. Thus regardless of whether the
aggregation function causes $\uppermech$ running VCG to give $G^1$
either $r$ or $s$, there always exists a profitable non-truthful
deviation by either agent $1^1$ or agent $2^1$ showing that $\mech$ is
never truthful in this setting.

\end{appendices}

\newpage

\bibliographystyle{plainurl}
\bibliography{bibliography}

\end{document}